%% file: PropertyLP.tex
\documentclass[12pt]{article}
\usepackage{amsmath}        
\usepackage{amsthm}
\usepackage{subfigure}
\usepackage{amsfonts}
\usepackage{makeidx}         
\usepackage{graphicx}        
\usepackage{multicol}        
\usepackage[bottom]{footmisc}
\usepackage{color}

\DeclareMathOperator{\Pareto}{Pareto}
\renewcommand{\Pr}{\mathbb{P}}

\newtheorem{lemma}{Lemma}

\begin{document}

\begin{center}{\large\sc On a property of the inequality curve $\lambda(p)$}
{   %
\\ 
} \vspace*{0.5cm} {Emanuele Taufer$^{{\rm{a}}}$, Flavio Santi$^{{\rm{a}}}$, Giuseppe Espa$^{{\rm{a}}}$, Maria Michela Dickson$^{{\rm{b}}}$,  }\\ \vspace*{0.5cm}
{\it 
$^{{\rm{a}}}$Department of Economics and Management, University of Trento - Italy 
}
{\it 
$^{{\rm{a}}}$Department of Statistics, University of Padua - Italy 
}
\vspace*{0.5cm} 

\end{center}
\abstract{The Zenga (1984) inequality curve $\lambda(p)$ is constant in $p$ for Type I Pareto distributions. We show that this property holds exactly only for the Pareto distribution and, asymptotically, for distributions with power tail with index $–\alpha$, with  $\alpha >1$. Exploiting these properties one can develop powerful tools to analyze and estimate the tail of a distribution. 
}


\baselineskip = 1.1\baselineskip

\section{Introduction}
\label{sec:1}

Let $X$ be a positive random variable with finite mean $\mu$, distribution function $F$, and probability density $f$. The inequality curve $\lambda(p)$ \cite{Zenga84} is defined as:
\begin{equation}\label{lambdap}
\lambda(p) = 1 - \frac{\log (1-Q(F^{-1}(p)))}{\log(1-p)}, \quad 0<p<1,
\end{equation}
where $F^{-1}(p)= \inf \{x\colon F(x) \geq p \}$ is the generalized inverse of $F$ and $Q(x) = \int_0^x t f(t) dt / \mu$ is the first incomplete moment. $Q$ can be defined as a function of $p$ via the Lorenz curve
\begin{equation}
L(p)=Q(F^{-1}(p)) =\frac{1}{\mu} \int_0^p F^{-1}(t) dt.
\end{equation}

The curve $\lambda(p)$ has the property of being constant for Type I Pareto distributions  and, as it will be shown, this property holds asymptotically for distributions $F$ satisfying 
\begin{equation}
\label{HT}
\bar{F} (x) = x^{-\alpha} L(x)\,,
\end{equation}
where $\bar F=1-F$, and $L(x)$ is a slowly varying function, that is $L(tx)/L(x) \to 1$ as $x \to \infty$, for any $t>0$. We will say that $\bar F$ is regularly varying (RV) at infinity with index $-\alpha$, denoted as $\bar F \in RV_{-\alpha}$. The parameter $\alpha>0$ is usually referred to as \emph{tail index}; alternatively, in the extreme value (EV) literature it is typical to refer to the EV index $\gamma>0$ with $\alpha = 1/\gamma$ (see e.g.~\cite{mcneil2005}).

These properties can be exploited in order to develop estimator an estimator of the tail index as well as a goodness of fit test for the Pareto distribution . 

Probably the most well-known estimator of the tail index is the Hill \cite{Hill1975} estimator, which exploits the $k$ upper order statistics. The Hill estimator may suffer from high bias and is heavily dependent on the choice of $k$ (see e.g.~\cite{embrechts1997}). It has been thoroughly studied and several generalization have appeared in the literature. For  recent review of estimation procedures for the tail index of a distribution see \cite{Gomes15}.

The approach to estimation proposed here, directly connected to the inequality curve $\lambda(p)$ has a nice graphical interpretation and could be used to develop graphical tools for tail analysis. Another graph-based method is to be found in \cite{Kratz96}, which exploits properties of the QQ-plot;  while a recent approach based on the asymptotic properties of the partition function, a moment statistic generally employed in the analysis of multi-fractality, has been introduced by \cite{Grahovacetal2015}; see also \cite{Jia2018} which analyzes the real part of the characteristic function at the origin.  For other related works see \cite{Leo06}, \cite{Leo13}, \cite{Mei15}.

\section{Properties of $\lambda(p)$}

For a Type I Pareto  distribution \cite[573~ff.]{johnson1995} with
\begin{equation}\label{ParetoF}
F(x)=1-(x/x_0)^{-\alpha}, \quad x \geq x_0
\end{equation}
it holds that $\lambda(p)=1/\alpha$, i.e.\ $\lambda(p)$ is constant in $p$. This is actually an if-and-only-if result, as we formalize in the following lemma:

\begin{lemma}
\label{lem:iff}
The curve $\lambda(p)$ defined in \eqref{lambdap} is constant in $p$ if, and only if, $F$ satisfies~\eqref{ParetoF}.
\end{lemma}

\begin{proof}
It is trivially verified that if $F$ satisfies~\eqref{ParetoF} then $\lambda(p)= 1/\alpha$. Suppose now that $\lambda(p) = k$, $p \in (0,1)$, where $k$ is some constant. Then it must hold that
$1-H(p) = (1-p)^k$
or equivalently, after some algebraic manipulation,
\begin{equation}
\int_0^p F^{-1}(u) du = \mu [ 1-(1-p)^k]
\end{equation}
Taking derivatives on both sides we have that
\begin{equation}
\frac{d}{dp} \int_0^p F^{-1}(u) du = \frac{d}{dp} \mu [ 1-(1-p)^k],
\end{equation}
which gets
$$
F^{-1}(p) = \mu k (1-p)^{k-1}
$$
from which, setting $x_p=F^{-1}(p)$, which implies $p=F(x_p)$, it follows that, after some further elementary manipulations,
$$
\left( \frac{x_p}{\mu k} \right)^{1/(k-1)} = 1-F(x_p).
$$
Setting $1/(k-1)=-\alpha$, properly normalized, the above $F$ follows ~\eqref{ParetoF},
\end{proof}

See \cite{Zenga84} for a detailed analysis and calculations of $\lambda(p)$ for other probability distributions. The following result can also be stated, asymptotically for the case where  $\bar F$ satisfies \eqref{HT} as it is stated in the next lemma.

For this purpose write 
\begin{equation}
\lambda(x)=1 -\frac{\log(1-Q(x))}{\log(1-F(x))}
\end{equation}

\begin{lemma}
\label{lem:power}
If $\bar F$ satisfies \eqref{HT}, then  $\lim_{x \to \infty} \lambda(x) = 1/\alpha$.
\end{lemma}
\begin{proof}
Assume \eqref{HT}, since $\bar{F}(x)=\int_x^\infty f(t) dt$; by Karamata's theorem it follows that the density $f(x) = L(x) x^{-(\alpha+1)}$ as $x \to \infty$; again , by Karamata's theorem:
$$
\mu(1-Q(x))=\int_x^\infty L(t) t^{1-(\alpha+1)} dt =L(x)  x^{-\alpha+1}, \quad x \to \infty.
$$
Then, as $x \to \infty$, 
\begin{equation}
\begin{split}
\lambda(x) &= 1- \frac{\log(\mu^{-1} L(x) x^{1-\alpha})}{\log(L(x) x^{-\alpha})}\\
           & = 1 - \frac{\log(L(x) x^{-\alpha})}{\log(L(x) x^{-\alpha})}+ \frac{1}{\alpha}\frac{\log(L(x) x)}{\log(L(x) x)} + \frac{1}{\alpha}\frac{\log(\mu)}{\log(L(x) x)} \\
					& = \frac{1}{\alpha} + O\left(\frac{1}{\log L(x)x}\right).
\end{split}
\end{equation}
\end{proof}
A tail property of Pareto type I distribution is worth of being noted. Let $X$ be a random variable distributed according to~\eqref{ParetoF} -- that is, $X\sim\Pareto(\alpha,x_0)$ --, the following property holds for any $x_1>x_2>x_0$:
\[
\Pr[X>x_1|X>x_2]=\left(\frac{x_1}{x_2}\right)^{-\alpha}\,,
\]
hence, the truncated random variable $(X|X>x_2)$ is distributed as $\Pareto(\alpha,x_2)$.

The implications of this property are twofold. Firstly, the truncated random variable is still distributed according to~\eqref{ParetoF}, thus Lemma~\ref{lem:iff} still applies. Secondly, the tail index $\alpha$ is the same both for original and for truncated random variable, thus function $\lambda(p)$ can be used for the estimation of $\alpha$ regardless of the truncation threshold~$x_2$.

The same property we have just outlined holds asymptotically for distribution functions satisfying \eqref{HT}.

Figure \ref{fig:f1} reports the empirical curve $\hat\lambda(p)$ as a function of $p$ for a Pareto distribution defined by~\eqref{ParetoF} with $\alpha=2$ and $x_0=1$, denoted with $\Pareto(2,1)$ and a Fr\'{e}chet distribution with $F(x)=\exp{(-x^{-\alpha})}$ for $x \geq 0$ and $\alpha=2$, denoted by Fr\'{e}chet(2) at different truncation thresholds. Note the remarkably regular behavior or the curves and the closeness to the theoretical form for the Fr\'{e}chet case already for low levels of truncation. 

\textit{{
\begin{figure}[!h]
	\centering
	\subfigure{
		\label{fig.1a}
		\includegraphics[scale=0.31]{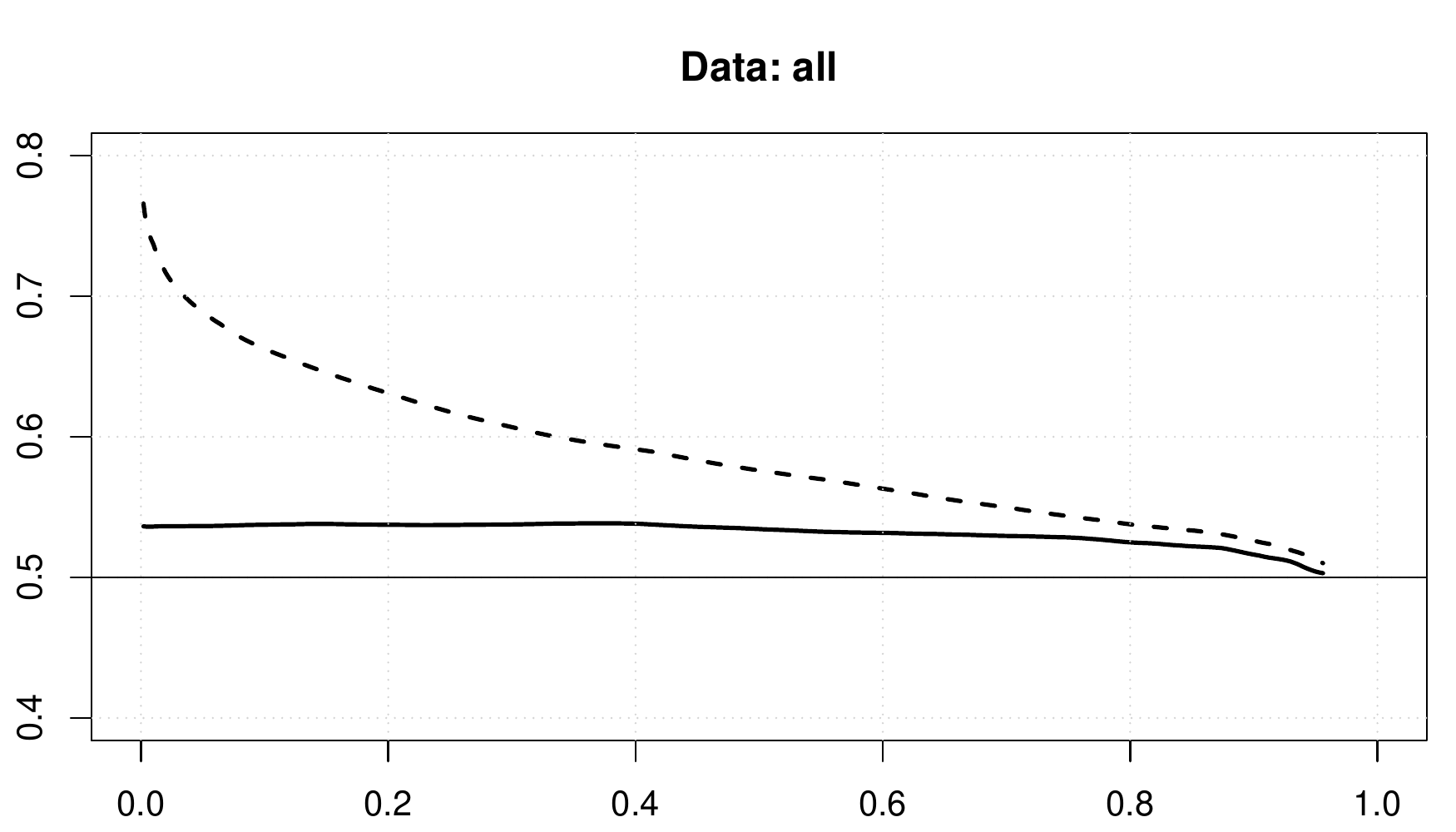}
	}
	\subfigure{
		\label{fig.1b}
		\includegraphics[scale=0.31]{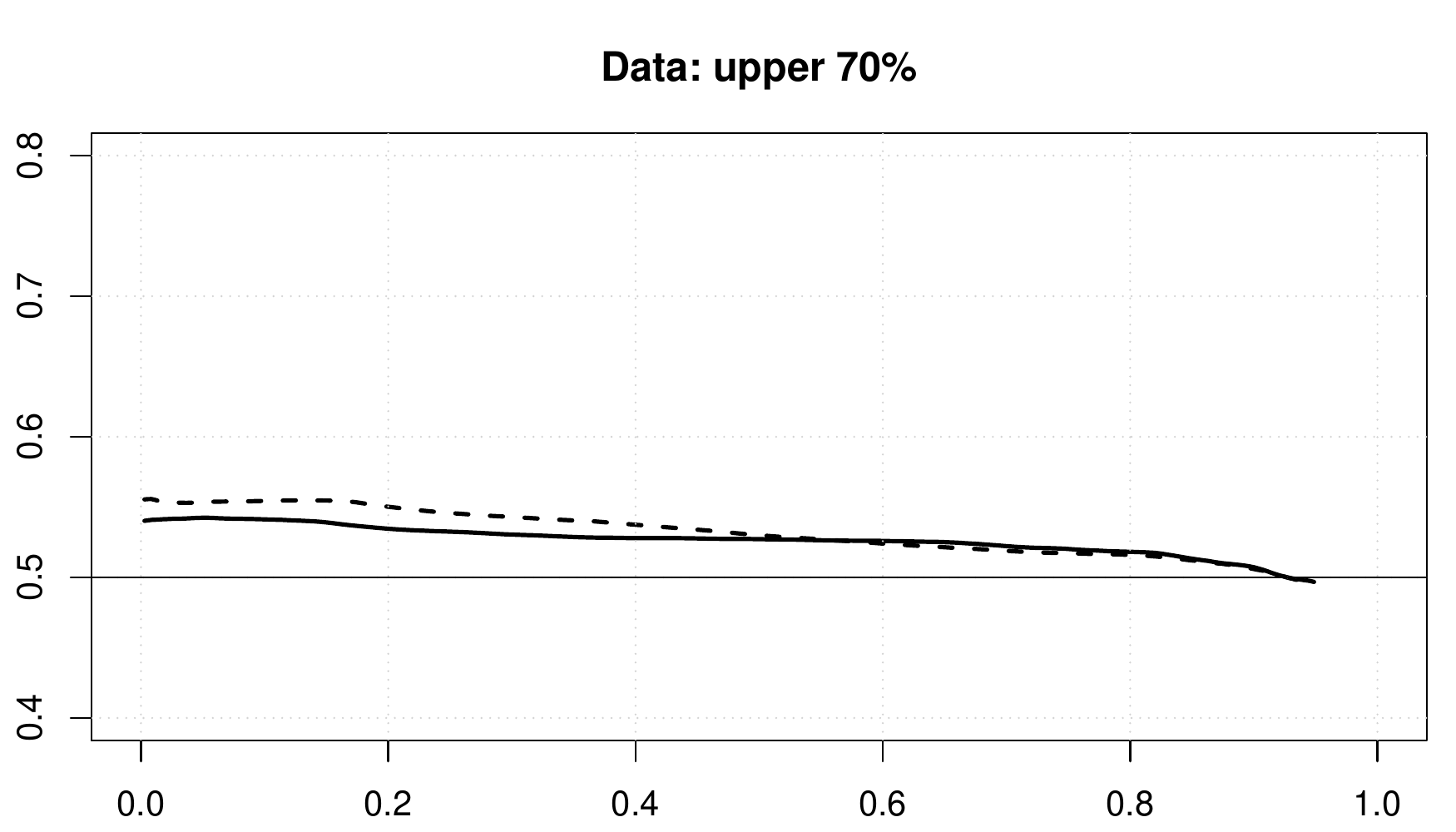}
	}
	\subfigure{
		\label{fig.1c}
		\includegraphics[scale=0.31]{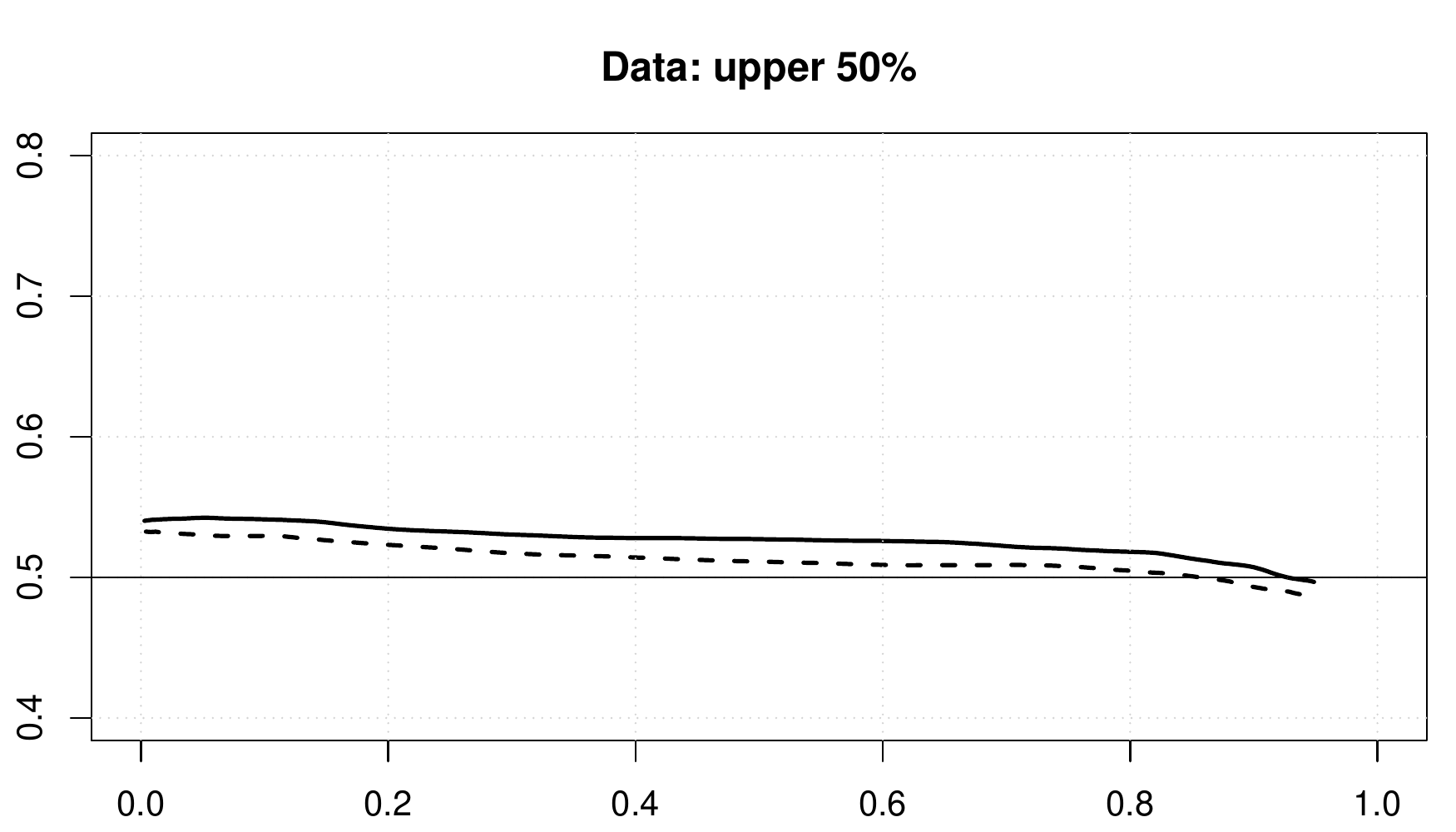}
	}
	\subfigure{
		\label{fig.1d}
		\includegraphics[scale=0.31]{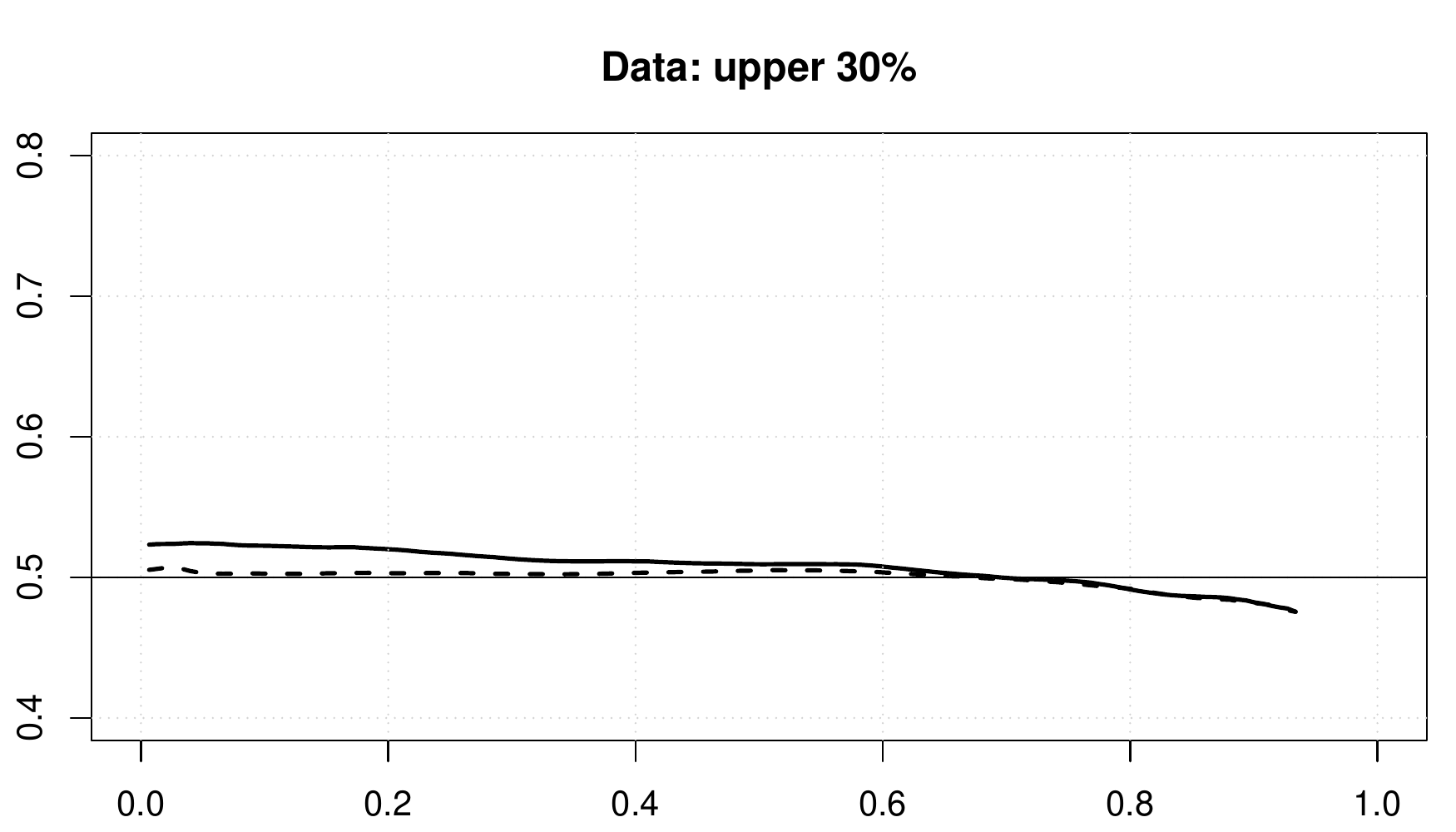}
	}
\caption{Plot of $\hat\lambda(p)$ and $p$ for $\Pareto(2,1)$ (solid line) and Fr\'{e}chet(2) (dashed line) at various levels of truncation. Sample size $n=500$. Horizontal line at $1/\alpha=0.5$} 
\label{fig:f1}
\end{figure}
}}

Let $X_{(1)}, \dots, X_{(n)}$ be the order statistics of the sample, $\mathbb{I}_{(A)}$ the indicator function of the event $A$. To estimate $\lambda(p)$, define the preliminary estimates
\begin{equation}
F_n(x)= \frac{1}{n} \sum_{i=1}^n \mathbb{I}_{(X_i \leq x)}  \qquad Q_n(x) = \frac{\sum_{i=1}^n X_i \mathbb{I}_{(X_i \leq x)}}{\sum_{i=1}^n X_{i}}
\end{equation}
Under the Glivenko-Cantelli theorem (see e.g.~\cite{resnick1999}) it holds that $F_n(x) \to F(x)$ almost surely and uniformly in $0<x<\infty$; under the assumption that $E(X) < \infty$, it holds that $Q_n(x) \to Q(x)$ almost surely and uniformly in $0<x<\infty$.  $F_n$ and $Q_n$ are both step functions with jumps at $X_{(1)}, \dots, X_{(n)}$. The jumps of $F_n$ are of size $1/n$ while the jumps of $Q_n$ are of size $X_{(i)}/T$ where $T=\sum_{i=1}^n X_{(i)}$. Define the empirical counterpart of $L$ as follows:
\begin{equation}\label{Lnp}
L_n(p) = Q_n( F_n^{-1}(p)) = \frac{\sum_{j=1}^i X_{(j)}}{T}, \quad \frac{i}{n} \leq p < \frac{i+1}{n}, \quad i=1, 2, \dots, n-1, 
\end{equation}
where $F_n^{-1}(p)= \inf \{x: F_n (x) \geq p \}$. To estimate $\alpha$ define

\begin{equation}
{\hat \lambda}_i = 1- \frac{\log(1-L_n(p_i))}{\log(1-p_i)}, \quad p_i = \frac{i}{n}, \quad i= 1, 2, \dots n-\left\lfloor \sqrt{n} \right\rfloor.
\end{equation}
and let $\hat \alpha = 1/ \bar{\lambda}$ where $\bar{\lambda}$ is the mean of the $\hat\lambda_i$'s. The choice of $i = 1, \dots , n - \left\lfloor \sqrt{n} \right\rfloor$ guarantees that  $\hat\lambda_i$ is consistent for $\lambda_i$ for each $p_i=i/n$ as $n \to \infty$.

\input{referenc-2}

\end{document}

%% file: referenc-2.tex
%
%
%